\documentclass{IEEEtran}
\usepackage{cite}
\usepackage{amsmath,amssymb,amsfonts}
\usepackage{algorithmic}
\usepackage{graphicx}
\usepackage{textcomp}
\usepackage{mathtools}
 \usepackage{xcolor}
\def\BibTeX{{\rm B\kern-.05em{\sc i\kern-.025em b}\kern-.08em
    T\kern-.1667em\lower.7ex\hbox{E}\kern-.125emX}}

\newtheorem{theorem}{Theorem}
\newtheorem{corollary}{Corollary}

\begin{document}
\title{Compressibility of Network Opinion and Spread States in the Laplacian-Eigenvector Basis}
\author{Sandip Roy and Mengran Xue 
\thanks{This work was partially supported by NSF Grant 1635184.}
\thanks{S. Roy is with School of Electrical Engineering and Computer Science, Washington State University, Pullman, WA 99164, USA. M. Xue is with Raytheon BBN Technologies, Columbia, MD. (emails: sandip@wsu.edu, morashu@gmail.com)}}

\maketitle

\begin{abstract}
 Opinion-evolution and spread processes on networks (e.g., infectious disease spread, opinion formation in social networks) are not only high dimensional but also volatile and multi-scale in nature.  In this study, we explore whether snapshot data from these processes can admit terse representations.  Specifically, using three case studies, we explore whether the data are compressible in the Laplacian-eigenvector basis, in the sense that each snapshot can be approximated well using a (possibly different) small set of basis vectors.  The first case study is concerned with a linear consensus model that is subject to a stochastic input at an unknown location; both empirical and formal analyses are used to characterize compressibility.  Second, compressibility of state snapshots for a stochastic voter model is assessed via an empirical study.  Finally, compressibility is studied for state-level daily COVID-19 positivity-rate data.  The three case studies indicate that state snapshots from opinion-evolution and spread processes allow terse representations, which nevertheless capture their rich propagative dynamics. 
\end{abstract}

\section{Introduction}

There has been an extensive effort to model
opinion-evolution and spread processes in networks, which spans social sciences, natural sciences, and engineering communities \cite{friedkin,prosurnikov,voter1,sars,ozdagler,pare,preciado}.  These studies postulate local update rules for agents' opinions or statuses, and seek to characterize emergent network-wide properties of the opinion-evolution/spread (e.g., attractivity to a manifold, amplification rate, settling time).  In parallel, empirical analyses of network opinion-evolution and spread data have been undertaken \cite{socialdata1,socialdata2,virusdata1}.  These efforts using field data clarify that opinion/spread processes exhibit sophisticated and volatile behaviors at multiple scales (e.g. stochastic and heterogeneous evolution rules, manipulative behaviors), which need to be accounted for in models and data analysis.  

The profound interconnectivity of modern human society, both with respect to cyber interactions and direct physical contact, is necessitating analysis of opinion and spread processes which have extremely high dimension.  For example, online social networks include tens-of-millions to billions of users, with tens-of-thousands of users sometimes communicating on a single topic during a span of a few minutes.  Modeling or data analysis for even a single trending topic therefore may require consideration of a high-dimensional process.  Analogously, as the COVID pandemic has highlighted, infectious disease spread and management may have global scale, while also requiring understanding of interactions at highly localized scales \cite{covid1}. 

The extreme dimensionality of opinion-evolution and spread processes motivates the development of terse descriptions, including reduced-order models and reduced representations of process data.  Indeed, reduced-order modeling of opinion-consensus dynamics has been pursued recently in the controls community \cite{consensusred1,consensusred2}.  In parallel, there have been numerous studies on reduced-dimension representation of social-network data, using techniques such as principal component analysis and factor analysis \cite{socialred1}.    These techniques are based on the concept that process states of interest are largely constrained to a fixed low-dimensional manifold within the full state space, which therefore allows model reduction and data projection.  

While it is appealing to approximate opinion and spread processes using fixed low-dimensional bases, the volatility and multi-scale dynamics of these processes may frustrate such approximations.  For example, the pattern of COVID-19 prevalences in geographical areas within the United States (e.g., states or counties) has shown considerable variation over time due to changing drivers and complex local spread processes; in consequence, it is unlikely that prevalence patterns throughout the epidemic can be approximated in a single low-dimensional basis.  Likewise, social interactions governing decision-making and voting processes, as well as opinion-diffusion in social networks, are highly stochastic and yield rich dynamics that cannot easily be projected on a common low-dimensional basis.

Even when network opinion and spread processes are not representable in a fixed low-dimensional basis, the state at each time may have considerable structure, exhibiting correlation, periodicity, or other features.  Moreover, the structure in the state may have a close connection with the graph topology of the network.  The purpose of this work is to examine whether state snapshots admit terse representations which capture this structure, while allowing for the inherent variability and sophisticated multi-scale dynamics of the processes.  

The notion of {\em compressibility}, which was developed in the signal-processing community for sparse signal reconstruction,  provides an interesting framework for terse representation of network process snapshots \cite{compressive1,compressive2}. In this literature, a signal ensemble is viewed as compressible, if it admits a sparse approximation in a particular full-dimensional basis -- however, one that may differ from one ensemble member to another.  For example, a class of images is compressible in the Fourier basis, if each image of this type has only a few frequency components, albeit perhaps different ones.  
In this study, we explore whether the state snapshots of network opinion and spread process are compressible, in the sense that each snapshot admits a (possibly-different) sparse representation in a particular full-dimensional basis.    We particularly focus on whether these processes are compressible in the {\em Laplacian-eigenvector basis}, defined by the right eigenvectors of the Laplacian matrix associated with the network's graph.

The article is focused around three case studies, which are used to explore whether model-generated and also field data of opinion/spread processes are compressible in the Laplacian-eigenvector basis.  To describe the case studies, we first define compressibility notions for network opinion/spread processes, and develop the Laplacian-eigenvector basis (Section II).  Then, in the first case study, we characterize compressibility of an opinion-consensus process that is subject to a stochastic input, using both empirical analyses of model-generated data and some formal analyses (Section III).  The second case study is focused on empirical analysis of compressibility for data from a stochastic voter model (Section IV).  Finally, in Section V, compressibility of COVID-19 field data is explored.  Outcomes of the cases studies and possible applications of compressibility are briefly summarized in Section VI.

\section{Definitions and Aims}


Compressibility of snapshot data (i.e., data at a particular time point) from a network opinion-evolution or spread process is studied.  Specifically, a data vector ${\bf x} \in R^N$, which captures the state of a network process at a particular time, is considered.  Each entry $x_i$ of ${\bf x}$ represents the opinion or {\em status} of a {\em node} $i$ in the network.  Broadly, ${\bf x}$ is viewed as being drawn from an ensemble ${\cal X}$, which encompasses state snapshots at different times or for different instantiations.

A weighted digraph with $N$ vertices, corresponding to the $N$ nodes in the network, is used to represent interactions or influences in the network.  Specifically, an edge is drawn from vertex $i$ to vertex $j$ to indicate that the node $i$ directly influences node $j$ in the evolution of the opinion/spread state.  For an edge from vertex $i$ to vertex $j$, a weight $w_{ij}>0$ is assigned which indicates the strength of the interaction.  

Exact and approximate notions of compressibility are defined for individual state snapshot data ${\bf x}$ and for the data ensemble ${\cal X}$, following on the definitions used in the compressive sensing literature \cite{compressive1,compressive2,compressive3,compressive4}.
These notions are defined with respect to a specific basis for $R^N$,  which we specify as the columns of a real $N \times N$ matrix $\Phi$.  The state snapshot ${\bf x}$ is defined to be exactly $K$-compressible with respect to $\Phi$, if the state snapshot ${\bf x}$ can be expressed as ${\bf x}=\Phi {\bf s}$, where at most $K$ entries of the real vector ${\bf s}$. If each vector ${\bf x} \in {\cal X}$ is exactly $K$-compressible, then the ensemble ${\cal X}$  is also referred to as exactly $K$-compressible.  

In many circumstances, exact compressibility is not achieved, but each state snapshot can be approximated well as a sparse combination of basis vectors.  The accuracy of a sparse approximation can naturally be measured in terms of the expected signal energy fraction captured by the approximation. In particular, the energy fraction $F$ of a $K$-sparse approximation $\overline{\bf x}= \Phi {\bf s}$ of ${\bf x}$, where ${\bf s}$ has $K$ non-zero entries, is defined as:
\begin{equation}
F=1-\frac{||{\bf x}-\overline{\bf x}||_2^2}{||{\bf x}||_2^2},
\end{equation}
where the subscript indicates the $2$-norm of the vector.

Also of interest is the maximum energy fraction among all $K$-sparse approximations:
\begin{equation}
    F^*=\max_{{\bf s} \, s.t. ||{\bf s}||_0 = K} F,
\end{equation}
where $||{\bf s}||_0$ indicates the number of non-zero entries.  We refer to $F^*$ as the optimal energy fraction for a $K$-sparse approximation, and the argument ${\bf s}={\bf s}^*$ that optimizes the energy fraction as the optimizing component vector (where we use the term `component' because ${\bf s}$ contains the components in the basis directions forming the approximation).  In the case where $\Phi$ is an orthonormal basis, it is easy to see that the optimizing component vector ${\bf s}^*$ can be found by first computing ${\bf s}=\Phi^{-1} {\bf x}$, and then setting to zero all except the $K$ largest-magnitude entries in ${\bf s}$.  For the general case, this approximation is not necessarily the optimal one, however maintaining the large-magnitude entries provides a good approximation (with provable performance relative to the optimal if the angles between basis vectors are lower-bounded).  Finally, the average value of the optimal energy fraction over the ensemble ${\bf X}$ (provided that the ensemble is stochastic), and/or the extremal values over the ensemble, may be of interest.


The exact compressibility of  a signal, as well as the energy fraction captured by a sparse approximation, are dependent on the basis $\Phi$ used for compression.  Idealized models for many opinion and spread processes are defined by the Laplacian matrix associated with the network's graph, and hence it is natural to assess compressibility in bases obtained from the Laplacian matrix.  In particular, the natural responses of deterministic linear opinion and spread models are sometimes primarily governed by dominant modes, which are aligned with certain right eigenvectors of the Laplacian matrix.  Thus, one might expect that data from real-world spread and opinion processes, or data generated from stochastic models of opinion formation/spread, may be compressible in the Laplacian-eigenvector basis -- i.e., state snapshots may primarily have components along a few dominant eigenvector directions, albeit in a time-dependent or instance-dependent fashion.  With this motivation in mind, here we study compressibility of opinion and spread dynamics in the Laplacian-eigenvector basis.  

Formally, the $N\times N$ Laplacian matrix $L$ of the digraph $\Gamma$ is defined as follows.  The off-diagonal entry of $L$ at row $j$ and column $i$ is set to $-w_{ji}$ if the graph has an edge from vertex $i$ to vertex $j$, and is set to $0$ otherwise.  Meanwhile, the diagonal entries are selected so that each row of the Laplacian matrix sums to $0$.  

The eigenvalues of the Laplacian matrix $L$ are known to lie in the closed right half plane, with all eigenvalues on the $j\omega$-axis located at the origin and non-defective.  We use the notation $0=\lambda_1,\hdots, \lambda_{N}$ for the eigenvalues of $L$.  We also define the matrix $V=\begin{bmatrix} {\bf v}_1, \hdots, {\bf v}_{N} \end{bmatrix}$ to contain the corresponding right eigenvectors (where each eigenvector is normalized to unit two-norm).  In the case where $L$ has complex eigenvalues, the corresponding eigenvectors (and generalized eigenvectors) are complex-conjugate pairs; in this case, we replace the conjugate vectors in $V$ with their real and imaginary parts. For the compressibility analysis, we focus on the case that the sparsifying basis is $\Phi=V$.

We note that compression of data in the Laplacian-eigenvector or graph-spectrum basis has recently been considered in the signal-processing literature for the purpose of compressive sensing \cite{graphcompress}, however to the best of our knowledge the compressibility of network process data has not been considered.

Our primary aim in this study is to explore whether the rich dynamics exhibited in network opinion-evolution and spread processes are compressible in the Laplacian-eigenvector basis.  Thus, in a deviation from the standard presentation in controls-engineering articles, we primarily focus on empirical assessment of compressibility in several examples or case studies.  In particular, three case studies are pursued: 1) data generated from a linear opinion-consensus model which is subject to a stochastic input at an unknown location; 2) data generated from a voter model; and 3) field data of daily COVID positivity rates in U.S. states.    An initial formal treatment of compressibility is also undertaken for the linear opinion-consensus model.

\section{Linear Opinion-Consensus Model}

Models for opinion consensus with stubborn or manipulative nodes (agents) have been widely studied \cite{ozdagler,pirani,koorehdavoudi}.  When opinion-consensus processes are impacted by such actors, the agents' states may not reach consensus, and in fact may continuously vary. 
Here, we consider a canonical discrete-time linear opinion consensus model, and augment it to include an additive stochastic input at an unknown node. Specifically, the consensus model is of the form:
\begin{equation}
    {\bf x}[k+1]=A{\bf x}[k]+Bu[k], \label{eq:gauss}
\end{equation}
where ${\bf x}[k] \in R^N$ contains the statuses of the $N$ nodes at time $k$, $A=[a_{ij}]$ is a row-stochastic matrix, $B$ is a $0--1$ indicator vector, and $u[k]$ is a zero-mean unit-variance Gaussian white noise process.  The entry of $B$ which is equal to $1$ (indicating the node where the input is applied), say the $z$th entry, is unspecified.  The model is representative of a network opinion-evolution process that is subject to manipulation at one node, which is unknown to network analysts.  Scenarios of this sort are common in social-network processes, where sources of manipulation or opinion-modification may be hidden.  Models of this form also can capture physical-world diffusion processes that are subject to unknown or stochastic drivers.  Many opinion-evolution and diffusion models of this type are high dimensional, and hence terse representations of noisy process data are desirable.

An $N$-vertex digraph $\Gamma$  is defined to capture the direct influence between nodes in the consensus model.  An edge is drawn from vertex $i$ to vertex $j$ ($i \neq j$) in $\Gamma$ if $a_{ji}>0$, and assigned a weight of $a_{ji}$.  We assume that the graph $\Gamma$ is strongly connected.  In this case, absent a driving input, the statuses of the nodes would reach a common value (consensus); however, with the stochastic input, the state varies with time. We note that the Laplacian matrix $L$ of the graph is closely related to the state matrix $A$ of the model, as $A=I-L$.

Our interest lies in determining the compressibility of state snapshots ${\bf x}={\bf x}[k]$ in the Laplacian-eigenvector basis.  Compressibility metrics (energy fractions) for individual snapshots, as well as for the stochastic ensemble at a particular time $k$, are studied.  An empirical analysis is undertaken for an example and then some initial formal analyses are pursued, which illustrate that the data is compressible.

\subsection{Empirical Analysis}

A network with $200$ nodes is considered.  The network's graph is formed by placing vertices in the unit plane, and connecting vertices within a certain radius (Figure \ref{fig:1}).  The edge weights are first set to a common value, and then the incoming weights into each vertex are scaled to sum to $0.8$ (hence the diagonal entries in $A$ equal 0.2).   A single node is selected randomly (with equal probability) as the location of the stochastic input.  The network instantiation is then simulated assuming a zero initial condition to generate state process data.

\begin{figure}[!htb]
\centering
\includegraphics[width=9.5cm,height=8cm]{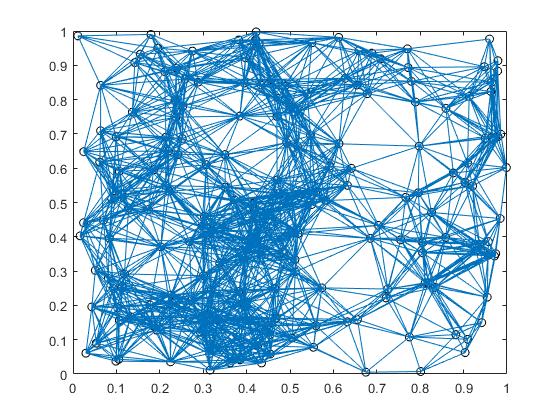}
\caption{A $200$ node graph is shown. Both linear-consensus dynamics (Section III.A) and voter dynamics (Section IV) were simulated on this graph. }
\label{fig:1}
\end{figure} 

Compressibility of the state snapshot ${\bf x}={\bf x}[k]$ at a particular time ($k=400$) in the Laplacian-eigenvector basis is considered.   Specifically, a $K$-sparse approximation of the state snapshot is determined by finding ${\bf s}=\Phi^{-1} {\bf x}[k]$, and setting to zero all but the $K$ largest entries.  The energy fraction $F$ is computed.  Also, the basis vectors (eigenvectors) with largest components in the $K$-sparse approximation are tabulated.

\begin{figure}[!htb]
\centering
\includegraphics[width=9cm,height=6.5cm]{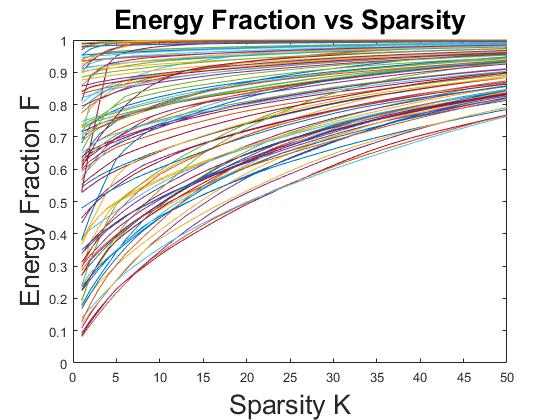}
 \includegraphics[width=9cm,height=6.5cm]{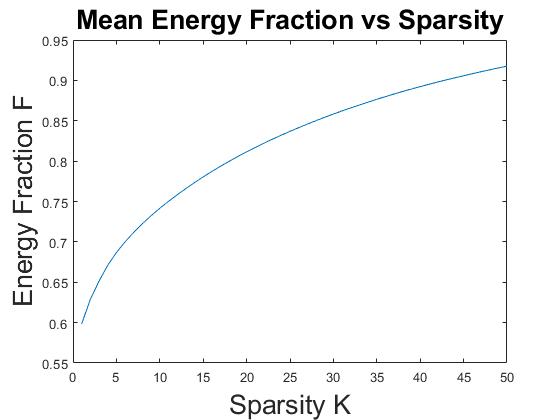}
\caption{For the linear consensus model, the energy fraction achieved by a $K$-sparse approximation is shown as a function of $K$, for each snapshot in the ensemble (top) and also averaged over the ensemble (bottom).  }
\label{fig:2}
\end{figure} 

Figure \ref{fig:2} shows the energy fraction for a $K$-sparse approximation of the state snapshot as a function of $K$, for an ensemble of simulations. The average energy fraction across the ensemble is also shown. The plots show that approximately $60\%$ of the energy is captured in the first basis vector on average, with $80\%$ of the total energy captured in the first 20 of the 200 basis vectors (equivalently, $50\%$ of the remaining energy beyond the first basis vector).  Approximately $90\%$ of the energy is captured in the first 40 basis vectors.  The energy fractions for the individual state snapshots show some variability, however for almost all ensemble members $50\%$ of the energy is captured using 20 basis vectors, and $70\%$ is captured using 40 basis vectors.

It is instructive to compare the dominant basis directions (the ones with the largest-energy components) for different ensemble members, both in the case that the input node is randomly varied and the case that the input node is held fixed.  As an example, for one simulation with the stochastic input at node $1$, we find that the indices of the first four dominant basis directions and corresponding Laplacian eigenvalues are: 1 ($\lambda=0$), 2 ($\lambda=0.04$), 19 ($\lambda=0.56)$, and 3 ($\lambda=0.05$).  Meanwhile, for another simulation with the same input location, they are 1 ($\lambda=0$),   167 ($\lambda=0.79$), 180 ($\lambda=0.80$), and 19 ($\lambda=0.56$).  Meanwhile, for a different input node, they are: 2 ($\lambda=0.04$), 3 ($\lambda=0.05$), 1 ($\lambda=0$), and 166 ($\lambda=0.79$).  The examples suggest that the dominant basis directions change among the ensemble members, even though a large energy fraction is captured within a few basis vectors.  Certain directions appear as dominant basis vectors with some frequency, particularly ones corresponding to eigenvalues near zero which are known to have network-wide span \cite{widearea}.  However, these basis directions are not always present, and other basis vectors appear as dominant ones depending on the ensemble member and input location.  The simulations suggest that state snapshots are compressible, but the dominant basis directions vary among the ensemble members.




\subsection{Statistical Analysis}

An initial statistical analysis of the linear opinion-consensus model (\ref{eq:gauss}) is undertaken, which gives some insight into compressibility of state snaphsots in the Laplacian-eigenvector basis.  The statistical analysis also suggests an alternate basis which allows for further compression, but requires some additional knowledge about the driving input.  

To simplify the presentation, we assume that eigenvalues of $A$ are real and simple.  This case encompasses state matrices that are symmetric or diagonally symmetrizable, which have been widely considered for network opinion-formation models.  The specialization is not essential for compressibility, however is simplifies the technical analysis.  For the analysis, we also assume that the matrix $A$ is aperiodic in addition to irreducible (which is a consequence of $\Gamma$ being strongly connected).

For the formal analysis, it is convenient to order the eigenvalues of $A$.  Specifically, we label the $N$ eigenvalues of $A$ as $1=\widehat{\lambda}_1 > \widehat{\lambda}_2 > \hdots >\widehat{\lambda}_{N} > -1$, where we have used the fact that $A$ has a simple strictly dominant eigenvalue at $1$ since it is irreducible and aperiodic.  The corresponding eigenvectors are labeled as $\widehat{\bf v}_1,\hdots, \widehat{\bf v}_{N}$.  

Importantly, the eigenvalues and eigenvectors of $A$ are closely related to those of the Laplacian matrix $L$ associated with the network's weighted digraph $\Gamma$.  Specifically, noticing that $A=I-L$, it follows that 
$\widehat{\lambda}_i=1-\lambda_i$, where the eigenvalues $\lambda_i$ of the Laplacian matrix satisfy $0 =\lambda_1 < \lambda_2 < \hdots < \lambda_{N} <2$.  Additionally, the eigenvectors of $A$ are identical to those of $L$, i.e. $\widehat{\bf v}_i={\bf v}_i$.  For the analysis, the vectors in the Laplacian-eigenvector basis are ordered in a commensurate fashion, i.e. $V=\begin{bmatrix} {\bf v}_1 & \hdots & {\bf v}_{N} \end{bmatrix}$.  Also, the matrix $W=[w_{ij}]$ is defined as $W=V^{-1}$.  The rows of $W$ are the left eigenvectors of $L$, which are also the left eigenvectors of $A$.

To assess compressibility, it is convenient to express the state ${\bf x}[k]$ of the linear-consensus process in the Laplacian-eigenvector basis.  In general, the state can be written as ${\bf x}[k]=V {\bf s}[k]$, where the vector ${\bf s}[k]$ lists the components in each basis direction.  Compressibility is essentially related to the squared magnitudes of the entries in ${\bf s}[k]$, which are statistical quantities.  Hence, we are interested in characterizing the second moments of the vector ${\bf s}[k]$, particularly in the large-$k$ asymptote.  The following result gives an explicit expression for the statistics of ${\bf s}[k]$:

\begin{theorem}
\label{th:1} Consider the opinion-consensus model in the case where the eigenvalues of $A$ are real and simple.  Consider the vector ${\bf s}[k]=V^{-1} {\bf x}[k]$, which indicates the components of the time-$k$ state ${\bf x}[k]$ in the Laplacian-eigenector basis.  The vector ${\bf s}[k]$ is Gaussian, with zero mean.  Also, in the limit of large $k$, the second moment of ${\bf s}[k]$ approaches:
\begin{equation} 
\Sigma=E({\bf s}[k]{\bf s}[k]^T)=QCQ. \label{eq:cov}
\end{equation}
Here, $Q$ is a diagonal matrix with $i$th diagonal entry given by the entry $w_{iz}$ in the left-eigenvector matrix $W$, where $z$ is the (unknown) location of the stochastic input in the opinion-consensus model.  Meanwhile, the entries of the matrix $C=[c_{ij}]$ are the following:
\begin{eqnarray}
& & c_{11}=k \\
& & c_{1j}=c_{j1}=\frac{1}{\lambda_j} \quad j=2,\hdots, n \nonumber \\
& & c_{ij}=\frac{1}{1-(1-\lambda_i)(1-\lambda_j)} \quad i=2,\hdots, n, \, j=2, \hdots, n .\nonumber
\end{eqnarray}

\end{theorem}

\begin{proof}
Substituting ${\bf x}[k]=V {\bf s}[k]$ into the state equation for the linear-consensus model yields:
${\bf s}[k+1]=\Lambda {\bf s}[k]+{\bf h} u[k]$,
where $\Lambda=diag(\widehat{\lambda}_i)$ and ${\bf h}=WB=\begin{bmatrix} w_{1z}\\ \vdots \\ w_{nz} \end{bmatrix}$.  From this transformed equation, ${\bf s}[k]$ can be expressed as:
${\bf s}[k]=\sum_{i=0}^{k-1} \Lambda^{k-1-i} {\bf h}{\bf u}[i]$.
From this expression, it is immediate that ${\bf s}[k]$ is zero mean.
Also, the second moment $\Sigma = E({\bf s}[k] {\bf s}[k]^T)$ can
be written as:
$\Sigma = \sum_{i=0}^{k-1} \sum_{j=0}^{k-1} \Lambda^{k-1-i}{\bf h}E({\bf u}[i]{\bf u}[i]^T ) {\bf h}^T \Lambda^{k-1-j}$.
Since ${\bf u}[k]$ is a zero-mean unit variance white process, $\Sigma$ can be further simplified to:
$\Sigma=\sum_{l=0}^{k-1} \Lambda^{k-1-l}{\bf h}{\bf h}^T \Lambda^{k-1-l}$.
Noticing that $\Lambda^{k-1-l}$ is diagonal, $\Sigma$ can be further simplified to: $\Sigma=Q CQ$, 
where $Q$ is defined in the theorem statement, and C=
$\sum_{l=0}^{k-1}\begin{bmatrix} \widehat{\lambda}_1^{k-1-l} \\ \vdots \\ \widehat{\lambda}_n^{k-1-l} \end{bmatrix} \begin{bmatrix} \widehat{\lambda}_1^{k-1-l} \\ \vdots \\ \widehat{\lambda}_n^{k-1-l} \end{bmatrix}^T$.  The entry of $C$ at row $i$ and column $j$ can be further simplified as 
$\sum_{l=0}^{k-1} (\widehat{\lambda}_i \widehat{\lambda}_j)^{k-1-l}$.  Noting that $\widehat{\lambda}_1=1$ while $\widehat{\lambda}_2,\hdots, \widehat{\lambda}_n$ are strictly within the unit circle, and considering large $k$, the expressions for the entries in $C$ in the theorem statement are recovered. $\blacksquare$

\end{proof}

Theorem \ref{th:1} shows that the components of the time-$k$ state ${\bf s}[k]$ in the Laplacian-eigenvector basis are zero-mean Gaussian random variables, with asymptotic covariance matrix given by Equation \ref{eq:cov}.  We are interested in  whether the energy in ${\bf s}[k]$ is concentrated in a small number $K$ of the largest entries, which would indicate compressibility.


The asymptotic expression for the second moment of ${\bf s}[k]$ gives insight into why the energy in ${\bf s}[k]$ is often concentrated in a small number of entries.  Specifically, the diagonal entries in the second-moment $\Sigma$ indicate the expected energies of each component in ${\bf s}[k]$.  Notice that the first diagonal entry in $\Sigma$, given by $w_{1z}^2 k$, is growing with $k$ while the remaining entries remain bounded.  Thus, the expected energy fraction along the first basis vector (the all-ones eigenvector ${\bf 1}$ of the Laplacian) approaches $1$ for large $k$.  Meanwhile, the remaining diagonal entries of $\Sigma$ take the form $\frac{w_{iz}^2}{1-(1-\lambda_i)^2}$.  These diagonal entries are much larger for $\lambda_i$ near $0$ or $2$ as compared to other values of $\lambda_i$, provided that the corresponding $w_{iz}^2$ are not vanishingly small.  For many network graphs, including the example considered above, the Laplacian has a relatively small set of eigenvalues near $0$, and the corresponding eigenvectors have wide support on the network.  Thus, the expected energy in the corresponding components of ${\bf s}[k]$ are amplified compared to the remaining components.  In an instantiation of the state, only a subset of these  components (basis directions) with large expected energy will have large amplitude in actuality.  Thus, the signal energy should be concentrated in a small, time-varying set of components in the Laplacian-eigenvector direction -- and thus the signal should be compressible.  If the location of the stochastic input is variable, the dominant components and their expected magnitudes will change (since $Q$ changes), but the state should remain compressible.

An exact analysis of the energy contained in the proposed $K$-sparse approximation is complicated, since it requires characterizing the order statistics of the entries in ${\bf s}[k]$. However, a simple lower bound on the energy can be found when the basis is orthonormal, based on the diagonal entries of $\Sigma$.  This bound is presented in the following  corollary:

\begin{corollary}
Consider the linear consensus model, and assume that the Laplacian matrix is symmetric (equivalently, the Laplacian-eigenvector basis $V$ is orthonormal).  Consider $K$-sparse approximations of ${\bf x}={\bf x}[k]$ in the Laplacian-eigenvector basis, for sufficiently large $k$.  Let $p(1),\hdots,p(K)$ be the indices for $K$ largest-magnitude diagonal entries of $\Sigma$.  Then the expected energy contained in the optimal $K$-sparse approximation is lower bounded by $\sum_{i=1}^{K} \sigma_{p(i),p(i)}$.  
\end{corollary}

\begin{proof}
For a particular time-$k$ snapshot ${\bf x}={\bf x}[k]$, consider the optimal $K$-sparse approximation $\overline{\bf s}$ in the Laplacian-eigenvector basis.  The energy contained in the approximation is equal to $\sum_{i=1}^K s_{q(i)}^2$, where $s_{q(1)},\hdots, s_{q(K)}$ are the $K$ largest-magnitude entries in ${\bf s}$.  The expected energy across the ensemble of snapshots is therefore given by 
$E(\sum_{i=1}^K s_{q(i)}^2)$.  However, for each snapshot, 
notice that $\sum_{i=1}^K s_{q(i)}^2 \ge \sum_{i=1}^K s_{p(i)}^2$.
The lower bound in the corollary statement follows.
\end{proof}



Theorem 1 also suggests an alternate basis which can allow sparse representation using fewer basis vectors.  In particular, previous work has shown that sparse representations of Gaussian data are obtained in an whitening basis under certain circumstances \cite{whitening}, i.e. one where the components are statistically independent.  The second-moment expression for ${\bf s}[k]$ in Theorem 1 shows that the Laplacian-eigenvector basis is not a whitening basis, since the entries in ${\bf s}[k]$ are correlated.  A whitening basis can be found using an eigenvalue decomposition of $\Sigma=QCQ$.  In particular, noticing that $\Sigma$ is positive definite and symmetric, it follows that $\Sigma$ can be decomposed as $\Sigma = V^* D {V^*}^{-1}$, where the orthonormal matrix $V^*$ contains the eigenvectors of $\Sigma$, and $D$ is a diagonal matrix with its eigenvalues.  Thus, it follows that the basis $\Phi=V^{*} V$ is a whitening basis for ${\bf x}[k]$ for sufficiently large $k$, with the second moment or covariance matrix in the new basis equal to $D$.    

The whitening basis has been applied to the $200$-node example network described in Section III.A.  In the new basis, a good sparse approximation is achieved with very few basis vectors.  In particular, a two basis vectors capture $89\%$ of the signal energy, and three basis vectors capture $99\%$ of the signal energy, and four basis vectors capture $99.9\%$ of te signal energy. 

In the example, the effectiveness of sparse approximation in the whitening basis results from the fact that the diagonal entries in $D$ (the component variances in the whitening basis) decay exponentially: when ordered by magnitude, each entry is a small fraction of the previous.  Indeed, this exponential decay in the component variances is a more general property of the whitening basis, for the class of models considered here.  This can be seen by first recognizing that the covariance matrix $\Sigma$ in the Laplacian-eigenvector basis is a controllability Gramian matrix for a single-input system.  The entries in $D$, which are the component variances in the whitening basis, are the eigenvalues of this Gramian.  However, the eigenvalues of Gramians for single-input systems have been shown to exhibit an exponential falloff for a wide range of system parameters, using properties of Cauchy matrices \cite{penzl,antoulas}.  These results explain why very sparse approximations can be obtained in the whitening basis.

While the whitening basis is appealing in enabling very efficient sparse approximations, there are practical challenges to its use in analyzing network opinion and spread data.  Importantly, constructing the basis requires knowing the location of the external input, which is often unrealistic in real-world scenarios. In addition, because the eigenvalues of $\Sigma$ have greatly varying orders of magnitude, computing the corresponding eigenvectors and hence the basis represents a computational challenge.  For these reasons, the Laplacian-eigenvector basis, though suboptimal, is practical for processing opinion/spread data in many real-world circumstances.

\section{Voter Model}

Stochastic automata network models, wherein nodes hold discrete-valued opinions or statuses which evolve through probabilistic interactions, are also widely used to represent opinion evolution.  In particular, voter models -- in which agents stochastically poll neighors to update their states -- have been used to represent various decision-making and algorithmic processes \cite{voter1,ozdagler}.  Reflecting real-world voting processes, Voter-model dynamics are intrinsicially random, which complicates reduced-order modeling and data analysis.  In addition, the models are sometimes used to represent individual opinions within large populations, and hence may be high dimensional.  For these reasons, terse representations of voter-model data are of interest.

Formally, we consider a (discrete-time) voter model with $N$ nodes, each of which has a status of either $0$ or $1$ at each time step.  The evolution of the state ${\bf x}[k]$ at each time $k$ can be described in two stages:

\noindent {\em Stage 1:}  The vector ${\bf y}[k]$ is computed as ${\bf y}[k]=A {\bf x}[k]$, where $A$ is an $N \times N$ row-stochastic matrix.  We note that each entry in ${\bf y}[k]$ is in $[0,1]$.

\noindent {\em Stage 2:}  Each entry in the next state ${\bf x}[k+1]$ is generated from ${\bf y}[k]$, as follows.  The entry $x_i[k+1]$ is set to $1$ with probability $y_i[k]$, and is set to $0$ otherwise, independently of the updates of the other statuses.

The voter model is entirely defined by the row-stochastic matrix $A$, in analogy with the linear consensus model.  We define the graph for the model in the same way as for the linear consensus model, since this graph captures direct influences between nodes.

An empirical study of the compressibility of state snapshots is undertaken for the voter model.  The $200$-vertex graph and stochastic matrix $A$ defined in Section III.A is also used for the voter model. Two nodes in the network are set to maintain statuses of $0$ and $1$, respectively.  The remaining nodes update their statuses as described above.  It can be seen that the nodes' statuses persistently fluctuate between $0$ and $1$.  The model state exhibits considerable temporal variability, but also  exhibits spatial correlation which should allow for compression.

\begin{figure}[!htb]
\centering
\includegraphics[width=8.5cm,height=6cm]{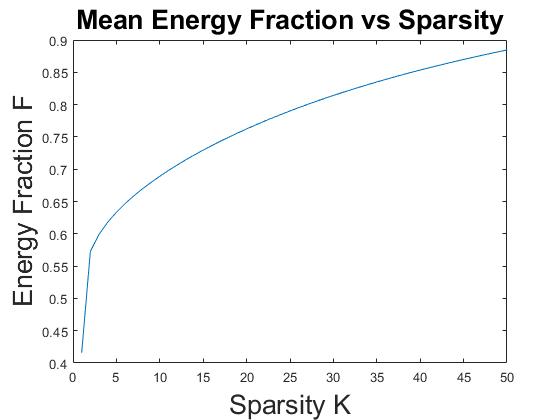}
\caption{For the voter model, the average energy fraction captured by a $K$-sparse approximation is plotted as a function of $K$.}
\label{fig:3}
\end{figure}

Compressibility of voter-model state snapshots in the Laplacian basis has been undertaken.  In particular, compressibility of the ensemble of voter-model states snapshots at a specific ($k=500$) has been studied.
The average energy fraction captured by an optimal $K$-sparse approximation  is plotted as a function of $K$ in \ref{fig:3}.  The approximation is seen to capture $75\%$ of the signal content for $K=20$, and $85 \%$ for $K=40$.  As with the linear consensus model, the dominant basis vectors vary significantly among the ensemble members, but a high energy fraction is achieved regardless of which basis vectors have large components.  While many different basis vectors may have large components, wide-area basis directions corresponding with small eigenvalues of the Laplacian are frequently dominant.

\begin{figure}[!htb]
\centering
\includegraphics[width=9cm,height=7cm]{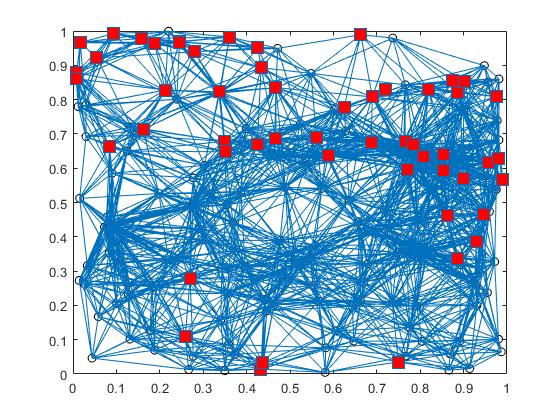}
\includegraphics[width=9cm,height=7cm]{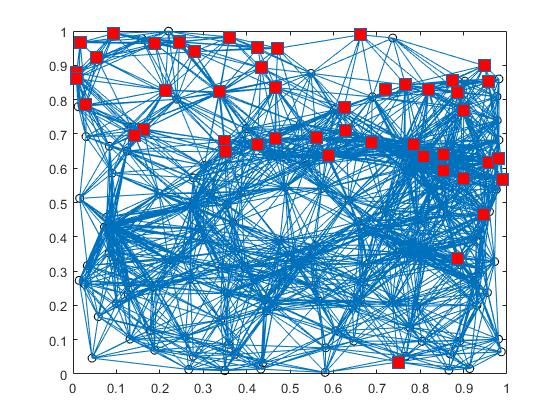}
\caption{Reconstruction of a discrete-valued voter model state using a $20$-sparse approximation is shown.  The original snapshot (top) and reconstruction (bottom) have an $89\%$ match.  (Red squares in the graph indicate statuses of `1'.}
\label{fig:4}
\end{figure}


It is instructive to see how effectively sparse approximations can recover the discrete state of the voter model.  To illustrate this, we have pursued reconstruction of the state snapshot from the sparse representation in one example, see Figure \ref{fig:4}.  Specifically, a sparse approximation has been computed using 20 basis vectors, and then translated to a discrete-state approximation by rounding.  The approximation matches the original voter model's state at 177 of the 200 nodes ($89\%$). It captures the state pattern well, while making errors at isolated nodes.  When 40 basis vectors are used, the approximation is correct at 193 of the 200 nodes.

\section{Daily COVID-19 Positivity Rates}

Compressibility of Coronavirus disease 2019 (COVID-2019) data was studied.  Specifically, daily state-level positivity rates for COVID-19 across the continguous United State (48 states + Washington DC) was examined, over a 250 day period.  The data were obtained from \cite{johnshopkins}.  The data for each day was viewed as a snapshot of the COVID spread state on a graph of the contiguous United States, with nodes representing the states and bidirectional edges indicating contiguous states.

\begin{figure}[!htb]
\centering
\includegraphics[width=8.5cm,height=6cm]{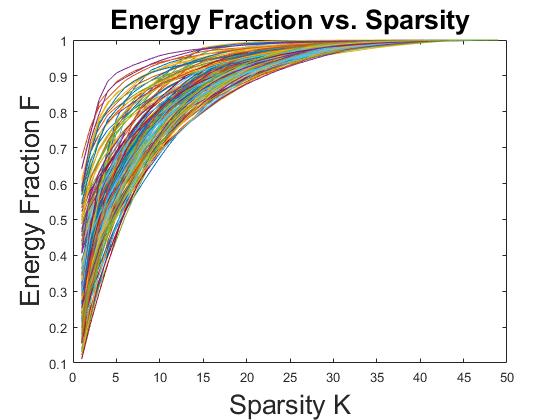}
\caption{For the COVID-19 positivity data, the optimal energy fraction achieved by a $K$-sparse approximation is shown as a function of $K$, for each day's data.}
\label{fig:5}
\end{figure}

The compressibility of the data in the Laplacian-eigenvector basis was examined.  As with the other examples, as $K$-sparse approximation was obtained by expressing the state snapshot in the (orthonormal) Laplacian eigenvector basis, and then maintaining the $K$ largest-magnitude entries while setting the remaining entries to $0$.  

Figure \ref{fig:5} shows the optimal energy fraction $F^*$ as a function of the sparsity $K$, for each state snapshot (day).  The plot shows that $10$ basis vectors are sufficient to capture $70-92\%$ of the energy (two-norm) in the state snapshot on every day, with an average optimal energy fraction of about $80\%$.  
Examination of the $K$-sparse approximations for each day shows that the dominant basis vectors are time-varying, showing gradual drifts  as well as periodic deviations caused by altered measurement protocols on weekends.

\begin{figure}[!htb]
\centering
\includegraphics[width=9.5cm,height=7cm]{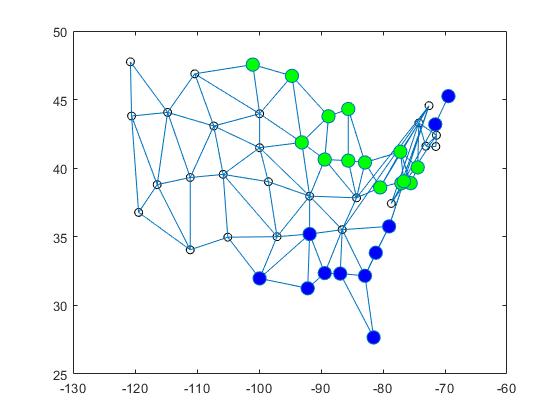}
\caption{A basis vector that has a large component on several days is plotted on top of the network graph. Blue and green circles are used to show strongly positive and strongly negative vector entries.  The basis vector  shows a gradient between the Norheastern and Southeastern United States.}
\label{fig:6}
\end{figure} 

Finally, we illustrate that plots of the dominant basis directions on top of the network graph can give intuition into the spread state.    In Figure \ref{fig:6}, one basis vector which has large-magnitude components on several days is plotted. The basis vector shows a gradient between the Northeastern and Southeastern United States, indicating that there is a broad difference in COVID positivity rates between the two regions on these days.

\section{Discussion and Possible Applications of Compressibility}

Our empirical and preliminary formal analyses indicate that snapshot data from network opinion-evolution and spread processes are compressible in the Laplacian-eigenvector basis. In this sense, network opinion and spread data are much like images and other scenes, but with the distinction that the compressive bases are tied to the network's topology.  The analogy suggests that alternate graph-related bases like graph wavelet bases may also yield sparse representations \cite{wavelet}.

The compressibility of opinion-evolution and spread process data may prove useful for several applications, such as:

1)  {\em State Reconstruction from Sparse Data}.   Reconstruction of opinion/spread processes from sparse measurements is of substantial interest.  Compressibility allows application of compressive-sensing techniques for reconstructing the state from sparse measurements.  In particular, for an $N$-dimensional state that is (approximately) $K$-sparse in a known basis, it is known that the full state can be recovered using on the order of $K \log (N/K)$ measurements under broad conditions \cite{compressive1}.  Furthermore, this can be done in a computationally appealing way using ${\cal L}^1$ optimization approaches.  

2)  {\em Sparse Control Design.}  The design of controls to manage opinion dynamics and mitigate spread is also of interest. Compressibility potentially may allow for sparse control designs, where actuation or influence is applied at only a small number of network nodes.

3) {\em Process Visualization.}  Compressibility enables approximation of an opinion/spread state in terms of a small number of Laplacian eigenvectors.  The Laplacian eigenvectors, when plotted on the network's graph, have wave-like shapes with different spatial frequencies.  Plotting the dominant eigenvector components for a particular state snapshot can thus assist in visualizing the snapshot.

\end{document}